\documentclass{article}

\usepackage{graphicx}
\usepackage{amsmath}
\usepackage{amssymb}
\usepackage{amscd}

\newtheorem{theo}{Theorem}[section]
\newtheorem{lemm}[theo]{Lemma}
\newtheorem{coro}[theo]{Corollary}
\newtheorem{prop}[theo]{Proposition}
\newtheorem{defi}[theo]{Definition}
\newenvironment{proof}{\textbf{Proof:}}{$\blacksquare$}
\newcommand{\EQ}{\begin{equation}}
\newcommand{\EN}{\end{equation}}
\newcommand{\w}{\mbox{wt}}

\newcommand{\Aut}{\mbox{\rm Aut}}

\newcommand{\wt}{\mbox{\rm wt}}

\newcommand{\bh}{{\bf h}}

\newcommand{\by}{{\bf y}}
\newcommand{\bx}{{\bf x}}
\newcommand{\bz}{{\bf z}}
\newcommand{\bv}{{\bf v}}
\newcommand{\bu}{{\bf u}}

\newcommand{\F}{\mathbb{F}}

\title{On linear completely regular
codes with covering radius $\rho=1$. Construction and
classification\thanks{This work has been partially supported by the Spanish MICINN Grants MTM2006-03250, TSI2006-14005-C02-01 and PCI2006-A7-0616, the AGAUR grant 2008PIV00050 and also by the Russian fund of fundamental researches 06-01-00226.}}

\author{J. Borges, J. Rif\`{a}\thanks{J. Borges and J. Rif\`{a}  are with the
Department of Information and Communications Engineering,
                             Universitat Aut\`{o}noma de Barcelona,
                             08193-Bellaterra, Spain
                             (email:~\{joaquim.borges,josep.rifa\}@autonoma.edu).}, V.A. Zinoviev\thanks{V.A. Zinoviev is with the Institute for Problems of Information
Transmission of the Russian Academy of Sciences, Bol'shoi Karetnyi
per. 19, GSP-4, Moscow, 127994, Russia (e-mail:\, zinov@iitp.ru).}}

\begin{document}

\maketitle

\begin{abstract}
Completely regular codes with covering radius $\rho=1$ must have
minimum distance $d\leq 3$. For $d=3$, such codes are perfect and
their parameters are well known. In this paper, the cases $d=1$ and
$d=2$ are studied and completely characterized when the codes are
linear. Moreover, it is proven that all these codes are completely
transitive.
\end{abstract}

\section{Introduction and Background}

Let $\F_q=GF(q)$ be the Galois Field with $q$ elements, where $q$ is
a prime power. $\F^n_q$ denotes the $n$-dimensional vector space
over $\F_q$. The all-zero vector in $\F^n_q$ is denoted by ${\mathbf
0}$. Let $\w(\bv)$ denote the {\em Hamming weight} of a vector $\bv
\in \F_q$ (i.e. the number of its nonzero positions), and $d(\bv,
\bu)=\w(\bv-\bu)$ denotes the {\em Hamming distance} between two
vectors $\bv$ and $\bu$. Given $\bv\in\F^n_q$, denote by $supp(\bv)$
the {\em support} of the vector $\bv$, that is, the set of
coordinate positions where $\bv$ has nonzero entries. We say that a
vector $\bu=(u_1,\ldots,u_n)\in\F^n_q$ {\em covers} a vector
$\bv=(v_1,\ldots,v_n)\in\F^n_q$ if $v_i\neq 0$ implies $v_i=u_i$.

A $q$-ary {\em code} $C$ of length $n$ is a subset of $\F^n_q$. If $C$ is a
$k$-dimensional linear subspace of $\F^n_q$, then $C$ is a {\em linear}
code, denoted by $[n,k,d]_q$, where $d$ is the {\em minimum distance} between any
pair of codewords.

Let $C$ be a $q$-ary code with minimum distance $d$, the {\em packing radius} of $C$ is
$$
e=\left\lfloor\frac{d-1}{2}\right\rfloor.
$$
Such a code is said to be an $e$-{\em error-correcting} code.

Given any
vector $\bv \in \F^n_q$, its {\em distance} to the code
$C$ is
\[
d(\bv,C)=\min_{\bx \in C}\{ d(\bv, \bx)\}
\]
and the {\em covering radius} of the code $C$ is
\[
\rho=\max_{\bv \in \F^n_q} \{d(\bv, C)\}.
\]

Clearly $e\leq\rho$ and $C$ is said to be {\em perfect} when $e=\rho$.

For any $\bx\in \F^n_q$, let $~D=C+\bx~$ be a {\em translate} of  $C$. The {\em weight}
$\w(D)$ of $D$ is the minimum weight of the codewords of $D$.

\begin{defi}\label{de:1.1}
A $q$-ary code $C$
is called {\em completely regular} if the weight
distribution of any translate $D$ of $C$
is uniquely defined by the weight of $D$.
\end{defi}

Equivalently, $C$ is completely regular if for all $\bx\in \F^n_q$ such that $d(\bx,C)=t$, the number of codewords
at distance $i$ ($0\leq i\leq n$) from $\bx$ depends only on $t$ and $i$.

Given a code $C$ with covering radius $\rho$, let $C(\rho)$ be the set of vectors at distance $\rho$
from $C$. The next statement can be found in \cite{Neum} for binary codes. For the non-binary case
it can be proven in similar way.

\begin{lemm}\label{CoveringSet}
If a $q$-ary code $C$ is completely regular with covering radius $\rho$, then $C(\rho)$
is also completely regular.
\end{lemm}

A {\em linear automorphism} of $\F^n_q$ is a coordinate permutation together with a product by
a nonzero scalar value at each position. Such an automorphism $\sigma$ can be represented by a
$n\times n$ monomial matrix $M$ such that $\bx M=\sigma(\bx)$, for all $\bx\in\F^n_q$. From now on,
if $C\subseteq \F^n_q$ is a linear code, the {\em full automorphism group} of $C$, denoted $\Aut(C)$,
is the group of linear automorphisms of $\F^n_q$ that leaves $C$ invariant. We say that
$\Aut(C)$ is {\em transitive} if it is transitive when acts on the set of weight one vectors of $\F^n_q$.

\begin{lemm}\label{Equivalents}
Let $C,D\subseteq\F^n_q$ be two linear equivalent codes, i.e. there
is a linear automorphism $\sigma$ of $\F^n_q$ such that
$D=\sigma(C)$. Then $\Aut(C)$ is transitive if and only if $\Aut(D)$
is transitive.
\end{lemm}

\begin{proof}
Notice that for all $g\in \Aut(C)$, $\sigma g \sigma^{-1}\in \Aut(D)$.
Assume that $\Aut(C)$ is transitive. Let $\bx$ and $\by$ be weight one vectors, we
want to find $\delta\in \Aut(D)$ such that $\delta(\bx)=\by$. Let
$\tau\in \Aut(C)$ such that
$\tau(\sigma^{-1}(\bx))=\sigma^{-1}(\by)$, then
$\sigma\tau\sigma^{-1}(\bx)=\by$ and $\sigma\tau\sigma^{-1}\in
\Aut(D)$. The statement then follows reversing the roles of $C$ and
$D$.
\end{proof}

For a linear code $C$, the group $\Aut(C)$ acts on the set of cosets of $C$ in the
following way: for all $\phi \in \Aut(C)$ and for every vector $\bv
\in \F^n_q$ we have $\phi(\bv + C) = \phi(\bv) + C$.

In \cite{giudici} and \cite{Sole} the following definition has
been given for the case of linear codes.

\begin{defi}
Let $C$ be a $q$-ary linear code with covering radius
$\rho$. Then $C$ is completely transitive if $\Aut(C)$ has
$\rho +1$ orbits when acts on the cosets of $C$.
\end{defi}

Since two cosets in the same orbit should have the same weight
distribution, it is clear that any completely transitive code is
completely regular. The following statement can be generalized for the case
$\rho> 1$ replacing transitivity by $\rho$-homogeneity \cite{Sole}. Here, we are only
interested in the case $\rho=1$.

\begin{lemm}\label{TransitiveCT}
Let $C$ be a $[n,k,d]_q$ code with covering radius $\rho=1$. If
$\Aut(C)$ is transitive, then $C$ is completely transitive.
\end{lemm}

\begin{proof}
Obvious, since all cosets of $C$, different of $C$, have leaders of weight 1. Thus, all such cosets are
in the same orbit.
\end{proof}

\medskip

It has been conjectured~\cite{Neum} for a long time that if $C$ is a
completely regular code and $|C|>2$, then $e \leq 3$. For the
special case of binary linear completely transitive codes
\cite{Sole}, the problem of existence is solved: it is proven in
\cite{Bor1,Bor2} that for $e \geq 4$ such nontrivial codes do not
exist. The conjecture is also proven for the case of perfect codes
($e=\rho$) \cite{Tiet,Zino} and quasi-perfect ($e+1=\rho$) uniformly
packed codes \cite{Goet,Tilb}, defined and studied also in
\cite{Bas1,Sem2}.

When $e\leq 3$, there are well known completely regular codes and,
recently, we have presented new constructions of binary and
non-binary completely regular codes \cite{Bor3,Rif1,Rif2}. However,
there does not exist a general classification of completely regular
codes with $e\leq 3$. In this paper we consider $q$-ary linear
completely regular codes with $\rho=1$.
A surprising fact is that to characterize all linear completely
regular codes with $\rho = 1$ we need only three constructions
($q$-repeated code construction, direct construction and Kronecker
product construction).

The paper is organized as follows. In Section 2 we present the $q$
times repeating construction to obtain linear or nonlinear $q$-ary
completely regular codes with $d=1$. In Section 3, we give a direct
construction to obtain $q$-ary linear completely regular codes with
$\rho=1$ and $d\in\{1,2\}$, we also introduce the Kronecker product
of matrices as an important tool to characterize $q$-ary linear
completely regular codes with $\rho=1$ and, finally, we show that
all such completely regular codes are completely transitive too.

\section{Completely regular codes with $d=1$ and the $q$-repeated code construction}

We start with a first example of family of completely regular codes with minimum distance $d=1$.

\begin{lemm}\label{Coveringd1}
Let $C$ be a perfect (binary or non-binary) code. Then $C(\rho)$ has
minimum distance 1.
\end{lemm}

\begin{proof}
Without loss of generality, we assume that ${\bf 0}\in C$. Let
$\bx\in C(\rho)$ with $\wt(\bx)=\rho$ and let $\bx'$ be a vector
such that $d(\bx,\bx')=1$ and $\wt(\bx')\geq \rho$. We claim that
$\bx'\in C(\rho)$ and then the minimum distance in $C(\rho)$ is 1.
Assume to the contrary that $\bx'\notin C(\rho)$, then clearly
$d(\bx',C)=\rho-1$. Notice also that a codeword $\by$ at distance
$\rho-1$ of $\bx'$ cannot be ${\bf 0}$. Hence we obtain a
contradiction because $\bx$ is at distance $\rho$ from more than one
codeword.
\end{proof}

As we have seen in Lemmas \ref{CoveringSet} and \ref{Coveringd1}, the covering set $C(\rho)$ of any perfect code is a
completely regular code with minimum distance $d=1$. In particular,
if $C$ is a single error-correcting code ($e=1$), then $C(\rho)$ is exactly
the complement of $C$. But these are not the only examples of
completely regular codes with $d=1$.

Let $C$ be a $[n,k,d]_q$ code. We construct the $q$-repeated code
$C'\subseteq \F^{n+1}_q$ of $C$ as follows: for any codeword
$\bx=(x_1,\ldots,x_n)\in C$, we have $q$ codewords in $C'$, namely
$$(0,x_1,\ldots,x_n), (1,x_1,\ldots,x_n),
\ldots,(q-1,x_1,\ldots,x_n).$$

\begin{lemm}\label{double}
Let $C$ be a $[n,k,d]_q$ code and let $C'\subseteq \F^{n+1}_q$ be
its $q$-repeated code. Let $\bx=(x_1,\ldots,x_n)\in\F^n_q$ be a
vector at distance $i$ from $\alpha_i$ codewords in $C$ and at
distance $i-1$ from $\alpha_{i-1}$ codewords in $C$. Then, any
vector of the form $\bx'=(x_0,x_1,\ldots,x_n)$ is at distance $i$
from exactly $\alpha_i+(q-1)\alpha_{i-1}$ codewords in $C'$.
\end{lemm}

\begin{proof}
For any codeword $\bz=(z_1,\ldots,z_n)\in C$ such that
$d(\bz,\bx)=i$ we have that $\bz'=(x_0,z_1,\ldots,z_n)$ is in $C'$
and $d(\bz',\bx')=i$. Moreover, for any codeword
$\by=(y_1,\ldots,y_n)\in C$ such that $d(\by,\bx)=i-1$, we have that
the $q-1$ vectors of the form $(y_0,y_1,\ldots,y_n)$ with $y_0\neq
x_0$ are codewords in $C'$ and they are at distance $i$ from $\bx'$.
It is clear that there are no more codewords in $C'$ at distance $i$
from $\bx'$.
\end{proof}

\begin{theo}\label{theo:3.1} $(q$-Repeated code construction$)$
Let $C$ be a $[n,k,d]_q$ code with covering radius $\rho$. Then the
$q$-repeated code $C'\subseteq \F^{n+1}_q$ has $\rho'=\rho$ and
minimum distance $d'=1$. Moreover $C'$ is completely regular if and
only if $C$ is completely regular.
\end{theo}

\begin{proof}
For any vector $\bx'=(x_0,x_1,\ldots,x_n)\in \F^{n+1}_q$, call
$\bx=(x_1,\ldots,x_n)\in \F^n_q$ the corresponding `reduced' vector.
Suppose that $\by=(y_1,\ldots,y_n)\in C$ is a codeword at minimum
distance from $\bx$. Then it is clear that
$\by'=(x_0,y_1,\ldots,y_n)$ is a codeword in $C'$ at minimum
distance from $\bx'$. Therefore $\rho=\rho'$.

Now, assume that $C$ is completely regular. For any vector
$\bx=(x_1,\ldots,x_n)\in \F_q^n$ at distance $t\leq \rho$ from $C$,
define $\alpha_i(t)$ as the number of codewords in $C$ at distance
$i$ from $\bx$ $(0\leq i \leq n)$. As $C$ is completely regular, we
know that $\alpha_i(t)$ does not depend on $\bx$, but just on $t$
and $i$. We want to see that for the vector
$\bx'=(x_0,x_1,\ldots,x_n)\in \F_q^{n+1}$, which is at distance $t$
form $C'$, we also have that the number of codewords in $C'$ at
distance $i$, say $\alpha'_i(t)$, depends only on $t$ and $i$. But
this is straightforward because using Lemma \ref{double} we have
$\alpha'_i(t)=\alpha_i(t)+(q-1)\alpha_{i-1}(t)$, for all
$i=0,\ldots,n$, and $\alpha'_{n+1}(t)=(q-1)\alpha_n(t)$.

Conversely, assume that $C$ is not completely regular. Let
$\bx,\by\in \F_q^n$ be such that $d(\bx,C)=d(\by,C)=t>0$ and let
$\alpha_{\bx,i}(t)$ (respectively $\alpha_{\by,i}(t)$) denote the
number of codewords at distance $i$ from $\bx$ (respect. $\by$), for
$0\leq i\leq n$. Since $C$ is not completely regular, we can select
$\bx$ and $\by$ such that $\alpha_{\bx,i}(t)\neq \alpha_{\by,i}(t)$
for some $i\geq t$. Let $i$ be the minimum possible such value
(possibly, $i=t$), that is
$\alpha_{\bx,i-1}(t)=\alpha_{\by,i-1}(t)$. Then, for the
$q$-repeated vectors $\bx'$ and $\by'$, we have
$\alpha'_{\bx',i}(t)\neq\alpha_{\by',i}(t)$ by Lemma \ref{double}.
Consequently, $C'$ is not completely regular.
\end{proof}

Hence, we can start with any completely regular code
and obtain an infinite family of completely regular codes with the
same covering radius. We remark that this construction is also valid for
nonlinear codes.

Conversely, for the linear case with $d=1$, we have the
following:

\begin{coro}\label{Casd1}
Let $C\neq\F^n_q$ be a $q$-ary linear code with minimum distance
$d=1$ and covering radius $\rho$. Then $C$ can be obtained using the
$q$-repeated code construction (repeating the process some number of
times) from a code $D$ which has minimum distance greater than one
and covering radius $\rho$. Moreover, $C$ is completely regular if
and only if $D$ is completely regular.
\end{coro}

\begin{proof}
Let $G$ be a generator matrix for $C$ containing all linear
independent codewords of weight 1. The desired code $D$ is then
obtained removing from $G$ all row vectors of weight 1 and the
resulting zero columns. As we have seen in Theorem \ref{theo:3.1},
the covering radius does not change and $C$ is completely regular if
and only if $D$ is completely regular.
\end{proof}

\section{Completely regular codes with $\rho=1$}

Since $e\leq\rho$, completely regular codes with $\rho=1$ must have
minimum distance $d\leq 3$. When $d=1$ we have seen, in the previous
section, that we can obtain these codes using the $q$-repeated
construction starting from codes with the same covering radius
$\rho=1$ and with minimum distance greater than 1. For $d=3$, we
have $e=\rho$ and these codes are perfect. Linear perfect codes with
$e=1$ are the well known Hamming codes.

Therefore, if $\rho=1$ the case to focus our interest is $d=2$. A
first example where we construct linear codes with these parameters,
$\rho=1$ and $d=2$, is given by the following theorem.

\begin{theo}\label{theo:4.1} $($Direct construction$)$
Let $C$ be a $[m+1,m,d]_q$ code defined by a generating matrix $G$,
\[
G~=~[I|\bh],
\]
where $I$ is the identity matrix of order $m$, and $\bh$ is an
arbitrary nonzero column vector from $\F^m_q$. Then, if
$\wt(\bh)<m$, the code $C$ is a completely regular code with
$d=\rho=1$. If $\wt(\bh)=m$, then $C$ is a completely regular code
with $d=2$ and $\rho=1$.
\end{theo}

\begin{proof}
Clearly, if $\wt(\bh)<m$, then the minimum distance of $C$ is 1 and if $\wt(\bh)=m$, then
the minimum distance is 2. A parity check matrix for $C$ is given by
$$
H=[-\bh^t | 1]
$$
and any pair of columns are linearly dependent. Hence $\rho=1$.


In order to see that $C$ is completely regular, we take a vector
$\bx$ at distance 1 from $C$ (or, the same, $\bx\notin C$) and we
prove that the number of codewords at distance 1 from $\bx$ is
always the same. Assume, without loss of generality, that
$\bx=(x_1,\ldots,x_{m+1})$ has weight 1. Let $w=\wt(\bh)$ and let
$x_i$ be the nonzero coordinate of $\bx$. First, we consider the
case $i< m+1$. The codewords at distance 1 from $\bx$ are ${\bf 0}$,
$x_i \bv^{(i)}$, where $\bv^{(i)}$ is the $i$-th row of $G$ (notice
that $\bv^{(i)}$ has weight 2, otherwise $\bx$ would be a codeword)
and the codewords of weight 2 with the value $x_i$ at the $i$-th
coordinate which are of the form: $\by^{(ij)}=x_i \bv^{(i)} +
\alpha_j \bv^{(j)}$ for all row vectors $\bv^{(j)}$ of weight 2
$(j\neq i)$, where $\alpha_j \in \F_q$ is taken such that the last
coordinate of $\by^{(ij)}$ is zero. Thus, we have $w+1$ codewords at
distance 1 from $\bx$. Finally, consider the case $i=m+1$. The
codewords at distance 1 from $\bx$ are ${\bf 0}$ and the $w$
codewords of the form $\by^{(j)}=\alpha_j \bv^{(j)}$, where
$\bv^{(j)}$ has weight 2 and $\alpha_j \in \F_q$ is taken such that
the last coordinate of $\by^{(j)}$ is $x_i$. Again, we obtain $w+1$
codewords at distance 1 from $\bx$.
\end{proof}

From~now on, our goal is to classify all the linear completely regular codes with $\rho=1$ and $d=2$.

We will begin by introducing the Kronecker product of matrices and showing that this tool will help us in the construction of linear completely regular codes with the required parameters.

\begin{defi}
The Kronecker product of two matrices $A=[a_{r,s}]$ and $B =[b_{i,j}]$ over
$\F_q$ is a new matrix $H = A \otimes B$ obtained by changing any element
$a_{r,s}$ in $A$ by the matrix $a_{r,s} B$.
\end{defi}

A {\em repetition code} is a $[n,1,n]_q$ code. In this paper, we
assume that such a repetition code has all codewords of the form
$(c,c,\ldots,c)$ for $c\in\F_q$.

\begin{lemm}\label{HammingTrans}
Let ${\cal H}$ be $[n,k,3]_q$ Hamming code. Then, $\Aut({\cal H})$
is transitive.
\end{lemm}

\begin{proof}
Let $G$ and $H$ be generator and parity check matrices,
respectively, for ${\cal H}$. Let $\bx$ and $\by$ be an arbitrary
pair of weight one vectors. We want to find a linear automorphism of
${\cal H}$ that sends $\bx$ to $\by$. It is straightforward to find
an invertible $(n-k)\times (n-k)$ matrix $K$, with entries in $\F_q$
and such that $KH\bx^t=H\by^t$. Since $H$ is a parity check matrix
of a Hamming code, there exists a monomial $n\times n$ matrix $M$
such that $KH=HM^t$. Since $H(M^t G^t)=KHG^t=0$, we have that $GM$
is also a generator matrix for ${\cal H}$. Thus, $M$ is the monomial
matrix associated to a linear automorphism $\phi\in\Aut({\cal H})$.

Now, $KH\bx^t=H\by^t$ implies $HM^t\bx^t=H\by^t$. As $M^t\bx^t$ and
$\by^t$ have weight one and $H$ has no repeated columns, we conclude
$\bx M=\by$ or, the same, $\phi(\bx)=\by$.
\end{proof}

\begin{theo}\label{theo:1}
Let $C$ be the linear code over $\F_q$ which has $H=A\otimes B$ as a
parity check matrix, where $A$ is a generator matrix for the
repetition $[n_a,1,n_a]_q$ code of length $n_a$ and $B$ is a parity
check matrix of a Hamming code with parameters $[n_b,k_b,3]_q$,
where $n_b=(q^{m_b}-1)/(q-1)$ and $k_b=n_b-m_b$.
\begin{enumerate}
\item [(i)] Code $C$ has length $n=n_a{\cdot}n_b$,
dimension $k= n - m_b$ and covering radius $\rho=1$.
\item[(ii)] If $n_a>1$, then the minimum distance of $C$ is $d=2$. If $n_a=1$, then $d=3$.
\item [(iii)] $\Aut(C)$ is transitive and, therefore, $C$ is a completely transitive code and a completely
regular code.
\end{enumerate}
\end{theo}

\begin{proof}  It is straightforward to check that the code $C$ has length
$n=n_a{\cdot}n_b$, dimension $k= n - m_b$ and covering radius $\rho=1$.

If $n_a=1$, then $C$ is a Hamming code and $d=3$. If $n_a>1$, then $H$ has repeated columns and $d=2$.

The matrix $H$ is of the form
$$
H=[B\;B\;\cdots B],
$$
where $B$ is a parity check matrix for a Hamming code ${\cal H}$. By
Lemma \ref{HammingTrans}, $\Aut({\cal H})$ is transitive on the set
of weight one vectors with support contained in the set of
coordinate positions of ${\cal H}$. Hence we have that $\Aut(C)$ is
transitive on each set of weight one vectors with support contained
in the set of $n_b$ coordinate positions corresponding to each
submatrix $B$. Now consider two vectors of weight one $\bx\in\F^n_q$
and $\by\in\F^n_q$, such that the nonzero entry of $\bx$ is
$x\in\F_q$ at position $i$ and the nonzero entry of $\by$ is $y$ at
position $j$ and assume that $i$ and $j$ are in different
$n_b$-sets, i.e. sets of cardinality $n_b$, of coordinate positions.
Let $\varphi\in\Aut(C)$ such that $\varphi(\bx)=\bx'$, where $\bx'$
has weight one with its nonzero entry equal to $y$ at position $i'$,
in the same $n_b$-set of coordinate positions, where the column
vector of $H$ in position $i'$ is the same that the column vector in
position $j$. Clearly, the transposition $\tau=(i',j)$ is a linear
automorphism of $C$. Thus, $\tau(\varphi(\bx))=\by$.

Therefore, we have proven that $\Aut(C)$ is transitive and, by Lemma
\ref{TransitiveCT}, $C$ is a completely transitive code and hence a
completely regular code.

\end{proof}

The following step is to prove that, vice versa, codes constructed in Theorem~\ref{theo:1}
are the unique linear completely regular codes with $d=2$ and $\rho=1$.

\begin{lemm}\label{lemm:5.2}
Let $C$ be a completely regular $[n,k,2]_q$ code with covering
radius $\rho=1$. Let $n_a$ be the number of codewords at distance 1
from any vector ${\bx}\notin C$. Then the following statements are
equivalent:
\begin{itemize}
\item[(i)] For any pair of coordinate positions $i$ and $j$, there
exists a codeword of weight 2 with support $\{i,j\}$.
\item[(ii)] $n_a = n$.
\item[(iii)] $C$ is a $q$-ary part of the whole space, i.e.
$|C|=q^{n-1}$.
\item[(iv)] Code $C$ has a generator matrix of the form
\[
G~=~[I|\bh],
\]
where $I$ is the identity matrix of order $n-1$, and $\bh$ is a
column vector of weight $n-1$ from $\F^{n-1}_q$.
\item[ (v)] Dual code of $C$  is equivalent to a repetition $[n,1,n]_q$ code.
\end{itemize}
\end{lemm}

\begin{proof}
Let ${\bx}\notin C$, without loss of generality we assume that
${\bx}$ has weight 1 and let $x_i$ be the nonzero coordinate of
${\bx}$. Then the codewords at distance 1 from ${\bx}$ are the
all-zero codeword and all codewords of weight 2 with $x_i$ at the
$i$-th coordinate. Such codewords have the remaining nonzero
coordinate in different places (otherwise $C$ would have codewords
of weight one). There are $n-1$ possible different places. Hence (i)
and (ii) are equivalent.


Define the following simple bipartite graph with vertices which are all
points of $\F_q^n$ and with edges, connecting the points of $C$
with the points $C(\rho) = \F_q^n \setminus C$, if these two
points are at distance one from each other. Count the number of
edges in two ways. From one side, any codeword of $C$ is at
distance $1$ from $(q-1)n$ points of $C(\rho)$. From the other
side, any point of $C(\rho)$ is at distance $1$ from $n_a$ points
of $C$. Since these numbers should be equal, we conclude that
\[
n(q-1)\,|C|~=~n_a\,|\F_q^n \setminus C|,
\]
which gives
\begin{equation}\label{SP-Condition}
(q-1)n=(q^{n-k}-1)n_a,
\end{equation}
where $k$ is the dimension of $C$. It is clear that $n_a = n$ if
and only if $k=n-1$. This gives the equivalence between (ii) and
(iii).

The equivalence between (iii) and (iv), and between (iv) and (v) are trivial.
\end{proof}

\begin{lemm}\label{lemm:5.3}
Let $C$ be a completely regular $[n,k,2]_q$ code with covering
radius $\rho=1$. Let $n_a$ be the number of codewords at distance
one from any vector not in $C$. If $k<n-1$, then the set of
coordinate positions $\{1,\ldots,n\}$ can be partitioned into
$n_a$-sets, $X_1,\ldots,X_{n/n_a}$, such that any codeword of weight
2 has its support contained in one of these sets.
\end{lemm}

\begin{proof}
First note that $n_a\geq 2$, otherwise $C$ would be a perfect
code with $d=2$ which does not exist. By Lemma \ref{lemm:5.2}, since $k<n-1$, we also have $n_a < n$
and clearly $n_a$ divides $n$ by (\ref{SP-Condition}).

Now, for any vector $\bu\notin C$ of weight 1, consider the union of the
supports of the $n_a -1$ codewords of weight 2 that cover $\bu$.
Denote by $X(\bu)$ such set of coordinate positions and note that
$|X(\bu)|=n_a$. Let $\bv$ be another vector of weight 1 such that
its support is not in $X(\bu)$. It suffices to prove that $X(\bu)$
and $X(\bv)$ are disjoin sets. Assume to the contrary that a
coordinate position $i$ belongs to $X(\bu) \cap X(\bv)$. This means
that there is a codeword $\bx$ of weight 2 covering $\bu$ and a
codeword $\by$ of weight 2 covering $\bv$ and $supp(\bx)\cap
supp(\by) = \{i\}$. Let $\by'$ be a multiple of $\by$ such that
$y'_i=x_i$. Then, the codeword $\bz=\bx - \by'$ covers $\bu$ but
$supp(\bz)\nsubseteq X(\bu)$ which is a contradiction.
\end{proof}

\begin{coro}\label{coro:1}
With the same hypothesis of Lemma \ref{lemm:5.3}, let $D_i$ be the
code that has the codewords of $C$ such that their supports are
contained in $X_i$ and deleting the coordinate positions outside of
$X_i$. Then, $D_i$ is a linear completely regular code of length
$n_a$, dimension $n_a -1$, minimum distance $d=2$, and
covering radius $\rho=1$. A generator matrix for $D_i$ is:
$$
G_i=[I | \bh],
$$
where $\bh$ is a column vector of weight $n_a-1$.
\end{coro}

\begin{proof}
For any $i=1,\ldots,n/n_a$, it is straightforward to see that $D_i$ is a
linear code of length $n_a$ and minimum distance $d=2$. Moreover,
let $Z$ be the set of weight two codewords covering some fixed
vector of weight one. Then $Z$ is a set of $n_a -1$ linear
independent codewords. Thus, by Theorem \ref{theo:4.1}, code $D_i$ is
completely regular with $\rho=1$.
\end{proof}

Now, it is clear that any linear completely regular code $C$ with
$d=2$ and $\rho=1$ can be `decomposed' into completely regular codes
$D_i$ of type `direct construction'. In order to complete the
classification we need the following technical results.

\begin{lemm}\label{lemm:5.4}
With the same hypothesis as in Lemma \ref{lemm:5.3}, let
$\bx=(x_1,\ldots,x_n)\in C$ and let $X_j$ be one of the sets as in
Lemma \ref{lemm:5.3}, such that $supp(\bx)\cap X_j \neq\emptyset$.
Then there exists a codeword $\bx'=(x'_1\ldots,x'_n)\in C$ which
coincides with $\bx$ in all positions outside of $X_j$,
such that $|supp(\bx')\cap X_j| \leq 1$, and where for the case
$|supp(\bx')\cap X_j| = 1$, the nonzero element of $\bx'$ occur
in any position of $X_j$, i.e. for any $i_j \in X_j$ there is a
such vector $\bx'$ with nonzero element in position $i_j$.
\end{lemm}

\begin{proof}
%
%
Let $\bx=(x_1,\ldots,x_n)\in C$ and let $X_j$ be such that
$supp(\bx)\cap X_j \neq\emptyset$. Now, adding codewords of weight
$2$ with support only in $X_j$ (see Lemma \ref{lemm:5.3}), from
$\bx$ we easily arrive to $\bx'$, which has either all zero
coordinates on $X_j$, or exactly one nonzero coordinate which might
be placed on any position of  $X_j$.
\end{proof}

\begin{prop}\label{prop:7}
With the same hypothesis as in Lemma \ref{lemm:5.3}, for each
$j=1,\ldots,n/n_a$, take and fix a coordinate position $i_j \in X_j$.
Let $D'$ be the code that
has all codewords in $C$ having their supports contained in
$I=\{i_1,\ldots,i_{n/n_a}\}$. Let $D$ be the code obtained from $D'$
by deleting all coordinates outside of $I$. Then $n/n_a \geq 3$ and
$D$ is a Hamming code of length $n/n_a$.
\end{prop}

\begin{proof}
Clearly $D$ is a linear code of length $n/ n_a$. By Lemma
\ref{lemm:5.3}, since we are assuming $k<n-1$, $D$ is not empty and
the minimum weight of $D$ is 3. Thus, we only need to prove that the
covering radius of $D$ is 1. Otherwise, assume that $\bv$ is a
vector (with coordinates in $I$) at distance 2 from $D$. Without
loss of generality, we can assume that $\bv$ has weight 2 with
$supp(\bv)=\{i_r,i_s\}$, ($i_r\in X_r, i_s\in X_s, r\neq s$). The
covering radius of $C$ is $\rho=1$, so we can take $\bx\in C$ at
distance one from $\bv'$, where $\bv'$ is the extension of vector
$\bv$ adding zeroes in all coordinate positions of
$\{1,\ldots,n\}\setminus I$. By Lemma \ref{lemm:5.3}, $\bx$ cannot
have neither weight 2 nor weight 1, since the minimum distance of
$C$ is 2. Thus, $\bx$ is a codeword of weight 3 with
$supp(\bx)=\{i_r,i_s,i\}$. Note that $i$ cannot be in $X_r$ or
$X_s$, otherwise, using Lemma \ref{lemm:5.4} we could obtain a
codeword of weight 2 with support $\{i_r,i_s\}$, contradicting Lemma
\ref{lemm:5.3}. We conclude that $n/n_a \geq 3$. Let $i\in X_t$,
where $r\neq t\neq s$. Again, using Lemma \ref{lemm:5.4}, we can
obtain a codeword $\bx'\in C$ such that
$supp(\bx')=\{i_r,i_s,i_t\}$, $x'_{i_r}=x_{i_r}$ and
$x'_{i_s}=x_{i_s}$. Clearly, $\bx'$ restricted to the $I$
coordinates is a codeword in $D$ of weight 3 and covers $\bv$.
Therefore $\bv$ is not at distance 2 from $D$.
\end{proof}

\begin{coro}\label{coro:8}
Let $C$ be a $[n,k,2]_q$ completely regular code with covering
radius $\rho=1$ and let $n_a$ be the number of codewords at
distance one from any vector not in $C$. Then, either $n_a=n$,
$k=n-1$ and $C$ has generator matrix:
$$
G=[G_1];
$$
or $C$ has generator matrix:
\begin{equation}\label{matgen}
G=\left[
    \begin{array}{ccc}
      G_1 & 0 & 0 \\
      0 & \ddots & 0 \\
      0 & 0 & G_{n/n_a} \\
      \hline
      M_1 & \cdots & M_{n/n_a} \\
    \end{array}
  \right],
\end{equation}
where $G_i$ is a generator matrix of a $[n_a,n_a -1,2]_q$
code (which is completely regular) for all $i=1,\ldots,n/ n_a$,
and $M_i$ has $n_a -1$ zero columns and one column $\bh_i$ such
that
$$
\left[\begin{array}{ccc}
    \bh_1 & \cdots & \bh_{n/ n_a} \\
    \end{array}\right]
$$
is a generator matrix of a Hamming code ${\cal H}$.
\end{coro}

\begin{proof}
We have already seen in Theorem \ref{theo:4.1} the case $n_a=n$, $k=n-1$.

Now, let $n_a < n$.
By Corollary \ref{coro:1} and Proposition \ref{prop:7}, it is clear
that code $C'$ generated by $G$ is a subcode of $C$. But, the
number of rows (which are all linear independent) of $G$ is:
$$
\frac{n}{n_a}\cdot (n_a -1) + dim({\cal H}).
$$
Since (\ref{SP-Condition}), the length of ${\cal H}$ is
$$
\frac{n}{n_a}=\frac{q^{n-k}-1}{q-1},
$$
the dimension of ${\cal H}$ is
$(n/n_a)-(n-k)$. Therefore
$$
dim(C')=\frac{n}{n_a}\cdot (n_a -1) +\frac{n}{n_a}-n+k=k.
$$
Hence, $dim(C')=dim(C)$ and consequently $C'=C$.
\end{proof}

\begin{prop}\label{prop:6.5}
Let $C$ be a $[n,k,2]_q$ completely regular code with covering
radius $\rho=1$. Let $n_a$ be the number of codewords at distance
one from any vector not in $C$. Let $A$ be a generator matrix of a repetition $[n_a,1,n_a]$-code
and let $B$ be a parity check matrix for a Hamming $q$-ary code of length $n_b=n/n_a$.
\begin{itemize}
\item[(i)] If $k=n-1$, then code $C$ is equivalent
to a code with parity check matrix $H=A$.
\item[(ii)] If $k<n-1$, then $C$ is equivalent
to a code with parity check matrix $H=A\otimes B$.
\end{itemize}
\end{prop}

\begin{proof}
If $k=n-1$, by Corollary \ref{coro:8}, code $C$ is given by a
generator matrix of a $[n_a,n_a-1,2]_q$ code. A parity check matrix
for an equivalent code to $C$ is the generator matrix of a
repetition $[n_a,1,n_a]$-code.

If $k<n-1$,
we can start with a generator matrix as in (\ref{matgen}).
Then, we multiply the first $n_a(n_a
-1)$ rows by appropriate values. After, we can multiply the columns
to obtain the following generator matrix:
\begin{equation}
G=\left[
    \begin{array}{ccc}
      G' & 0 & 0 \\
      0 & \ddots & 0 \\
      0 & 0 & G' \\
      \hline
      M_1 & \cdots & M_{n/n_a} \\
    \end{array}
  \right],
\end{equation}
where $G'$ is a $(n_a -1) \times n_a$ matrix
$$G'~=~[I|\bh].
$$
Up to equivalence, we can assume that $\bh$ has the value $q-1$ in
all its entries and $M_1,\ldots,M_{n/n_a}$ are as in Corollary
\ref{coro:8}. We also assume that the nonzero column of each $M_i$
is the first one.

Finally, we can permute the columns of the matrix
$$
H=[B\;B\;\cdots\;B]
$$
to obtain the matrix
$$
H'=[B_1\;B_2\;\cdots\;B_{n/n_a}],
$$
where $B_i$ has all its columns equal to the $i$-th column of $B$.
It is straightforward to see that $G$ and $H'$ are orthogonal matrices.
\end{proof}

Finally, we summarize the main result of this paper.

\begin{theo}
Let $C$ be a $[n,k,d]_q$ completely regular code with covering
radius $\rho=1$. Let $A$ be a generator matrix for the repetition
$[n_a,1,n_a]_q$ code of length $n_a$ and let $B$ be a parity check
matrix of a Hamming code with parameters $[n_b,k_b,3]_q$, where
$n=n_an_b$, $n_b=(q^{m_b}-1)/(q-1)$, $k_b=n_b-m_b$.
\begin{itemize}
\item[(i)] If $d=1$, then $C$ is the $q$-repeated code of a completely regular code $C'$ with covering
radius $\rho'=1$ and minimum distance $d'\in\{1,2\}$.
\item[(ii)] If $d=2$, then $n_a>1$ and $C$ is equivalent to a code with parity check matrix $H=A$ or $H=A\oplus B$.
\item[(iii)] If $d=3$, then $n_a=1$ and $C$ is a Hamming code and $H=B$ is a parity check matrix for $C$.
\item[(iv)] $C$ is a completely transitive code.
\end{itemize}
\end{theo}

\begin{proof}
We know that $d\in\{1,2,3\}$. We separate these three cases:
\begin{itemize}
\item[(i)] We have proven this statement in Corollary \ref{Casd1}.
\item[(ii)] This is proven in Proposition \ref{prop:6.5}.
\item[(iii)] Obvious, since $C$ is a perfect code.
\item[(iv)] If $d\in\{2,3\}$, by Proposition \ref{prop:6.5} and Theorem~\ref{theo:1}, $C$ is equivalent to a code $C'$ such that
$\Aut(C')$ is transitive. Thus, by Lemma \ref{Equivalents}, $\Aut(C)$ is transitive and, by Lemma \ref{TransitiveCT},
$C$ is completely transitive.

If $d=1$, then let $D$ be the `reduced' code, that is, the code
obtained from $C$ by doing the reverse operation of the $q$-repeated
code construction. Since the covering radius of $C$ and $D$ is 1, we
have that $C\neq \F^n_q$ and $D$ is a completely regular code with
$d>1$ by Theorem \ref{theo:3.1}. Therefore $D$ is a completely
transitive code. This means that we can choose a set of $q^{n-k}-1$
coset leaders of weight one such that they are in the same orbit of
$\Aut(D)$. But $C$ has the same number of cosets and we can choose
the same coset leaders. Since, clearly, $\Aut(D)\subseteq \Aut(C)$,
we have that these coset leaders are in the same orbit. Therefore,
all cosets different of $C$ are in the same orbit and $C$ is a
completely transitive code.
\end{itemize}
\end{proof}

\end{document}